\documentclass[a4paper,cleveref,USenglish,autoref, thm-restate]{lipics-v2021}

\usepackage[utf8]{inputenc}
\usepackage{url}
\usepackage{float}
\usepackage{amsmath}
\usepackage{hyperref}
\usepackage{graphicx}
\usepackage{markdown}
\usepackage{amsthm}
\usepackage{appendix}
\usepackage{xcolor}
\usepackage{enumitem}
\usepackage{tikz}
\usetikzlibrary{arrows}
\usetikzlibrary{patterns}
\usetikzlibrary{positioning}
\usepackage{multirow}
\usepackage{algorithm}
\usepackage{algorithmicx}
\usepackage[noend]{algpseudocode}
\usepackage{appendix}
\usepackage{tikz}
\usepackage{algorithm}
\usepackage[noend]{algpseudocode}
\usepackage{algorithmicx}
\usepackage{caption}
\usepackage{subcaption}
\usepackage{amsmath}
\usepackage{amsthm}
\usepackage{booktabs} 
\usepackage{amssymb}
\usepackage{xspace}
\usepackage{tikz-qtree}
\usetikzlibrary{calc}
\newcommand\payoff[1]{
  $\begin{pmatrix} #1 \end{pmatrix}$
}

\newcommand{\hash}{\textsf{VRF}_{sk}}
\newcommand{\vrf}[1]{\textsf{VRF}_{sk_{#1}}}
\newcommand{\verifyVRF}{\textsf{verifyVRF}_{pk}}
\newcommand{\Merklize}{\textsf{Merklize}}
\newcommand{\MerkleProof}{\textsf{MerkleProof}}
\newcommand{\verifyMerkleProof}{\textsf{verifyMerkleProof}}
\newcommand{\podgame}{\textsf{PoD-Game}\xspace}

\newcommand{\dcgame}{\textsf{DC-Game}\xspace}
\newcommand{\lcgame}{\textsf{LC-Game}\xspace}
\newcommand{\checkState}{\textsf{checkState}}
\newcommand{\true}{\textsf{true}}
\newcommand{\false}{\textsf{false}}
\newcommand{\todo}[1]{{\textcolor{red}{[TODO: #1]}}}
\newcommand{\validate}{\textsf{validate}}

\newcommand{\apply}{\textsf{apply}}
\newcommand{\hst}[1] {{\textcolor{green}{HT: #1}}}
\newcommand{\ps}[1] {{\textcolor{blue}{PS: #1}}}

\EventEditors{Rainer B\"{o}hme and Lucianna Kiffer}
\EventNoEds{2}
\EventLongTitle{6th Conference on Advances in Financial Technologies (AFT 2024)}
\EventShortTitle{AFT 2024}
\EventAcronym{AFT}
\EventYear{2024}
\EventDate{September 23--25, 2024}
\EventLocation{Vienna, Austria}
\EventLogo{}
\SeriesVolume{316}
\ArticleNo{6}
\title{Proof of Diligence: Cryptoeconomic Security for Rollups}

\author{Peiyao Sheng}{University of Illinois Urbana-Champaign \and Witness Chain}{psheng2@illinois.edu}{}{}
\author{Ranvir Rana}{Witness Chain}{ranvir.rana@kaleidoscope-blockchain.com}{}{}
\author{Senthil Bala}{Witness Chain}{senthil.bala@kaleidoscope-blockchain.com}{}{}
\author{ Himanshu Tyagi}{Witness Chain}{himanshu.tyagi@kaleidoscope-blockchain.com}{}{}
\author{Pramod Viswanath}{Princeton University \and Witness Chain}{pramod.viswanath@kaleidoscope-blockchain.com}{}{}
\authorrunning{Peiyao S et al.}
\Copyright{Peiyao Sheng, Ranvir Rana, Senthil Bala, Himanshu Tyagi, and Pramod Viswanath}
\relatedversiondetails[]{Full Version}{https://arxiv.org/abs/2402.07241}
\acknowledgements{This work was conducted when all authors were working at Witness Chain. }
\nolinenumbers
\begin{document}
\maketitle


\begin{abstract}
Layer 1 (L1) blockchains such as Ethereum are secured under an "honest supermajority of stake" assumption for a large pool of validators who verify each and every transaction on it. This high security comes at a scalability cost which not only effects the throughput of the blockchain but also results in high gas fees for executing transactions on chain. 
The most successful solution for this problem is provided by optimistic rollups, Layer 2 (L2) blockchains that execute transactions outside L1 but post the transaction data on L1. 

The security for such L2 chains is argued, informally, under the assumption that a set of nodes will check the transaction data posted on L1 and raise an alarm (a fraud proof) if faulty transactions are detected. However, all current deployments lack a proper incentive mechanism for ensuring that these nodes will do their job ``diligently'', and simply rely on a cursory incentive alignment argument for security. 

We solve this problem by introducing an incentivized watchtower network designed to serve as the first line of defense for rollups. Our main contribution is a ``Proof of Diligence'' protocol that requires watchtowers to continuously provide a proof that they have verified L2 assertions and get rewarded for the same. Proof of Diligence protocol includes a carefully-designed incentive mechanism that is provably secure when watchtowers are rational actors, under a mild rational independence assumption. 

Our proposed system is now live on Ethereum testnet. We deployed a watchtower network and implemented Proof of Diligence for multiple optimistic rollups. We extract execution as well as inclusion proofs for transactions as a part of the bounty. Each watchtower has minimal additional computational overhead beyond access to standard L1 and L2 RPC nodes. Our watchtower network comprises of 10 different (rationally independent) EigenLayer operators, secured using restaked Ethereum and spread across three different continents, watching two different optimistic rollups for Ethereum, providing them a decentralized and trustfree first line of defense. The watchtower network can be configured to watch the batches committed by sequencer on L1, providing an approximately 3 minute (cryptoeconomically secure) finality since the additional overhead for watching is very low. This is much lower than the finality delay in the current setup where it takes about 45 minutes for state assertions on L1, and hence will not delay the finality process on L1. 




\end{abstract}

\section{Introduction}
\paragraph*{L1 Scalability.} 
Blockchain scalability, via improving throughput, latency and transaction fees, is a crucial component of blockchain adoption~\cite{hafid2020scaling,monte2020scaling}. Improving the core L1 consensus protocols is a key direction to scalability~\cite{chen2016algorand,guo2020dumbo,diem,danezis2022narwhal};  as an example,  Ethereum upgraded consensus from the longest chain protocol in the proof-of-work (PoW) setting to the GHOST protocol~\cite{sompolinsky2015secure} with Casper finality gadget~\cite{casper} in the proof-of-stake (PoS) 
setting recently, improving its throughput potentially by three orders of magnitude and reducing 99.95\% energy demand~\cite{Ethereum2023Merge,Gfinity2022Ethereum2TPS}. Changes to L1  via upgrading the consensus protocol is a wholesale change (hard fork), requiring significant social consensus around the process and is time-consuming (e.g., Ethereum's upgrade to PoS from PoW took seven years~\cite{CoinTelegraph2020ETH2Upgrade}).  
However the scalability challenges have persisted with increased computational demands and data storage outstripping the upgraded capabilities, with the result that resulting gas fees have not seen any significant reduction. More scaling techniques that maintain equivalent decentralization, known as horizontal scaling, are necessary.

\paragraph*{L2 Economy.}  L2 solutions have emerged as a significant breakthrough, with potential to increase transaction processing speed, cut down costs, and enhance the overall capacity of blockchains. 
Rollups work by processing transactions outside the L1 chain and then posting the transaction data and state commitments back to it. 
Established rollup platforms such as  Optimism \cite{optimism} and Arbitrum \cite{arbitrum}  accommodate a vast ecosystem (e.g., Arbitrum L2 supported more transactions than Ethereum L1 itself in February 2023 \cite{CoinDesk2023Arbitrum}). At the same time, new entities \cite{base,opbnb,linea,layern} are entering the scene, introducing new infrastructure and optimizing their functionalities for specific applications. 
These new schemes differ from traditional rollups in their targeted functionalities and optimization techniques.  Moreover, driven by the demand for customizable and accessible layer 2 solutions, the concept of ``Rollups-as-a-Service'' is gaining traction \cite{caldera,conduit,eclipse,altlayer}, allowing a broader range of participants to create and utilize their own rollup strategies. 

\begin{figure*}
    \centering
    \includegraphics[width=\textwidth]{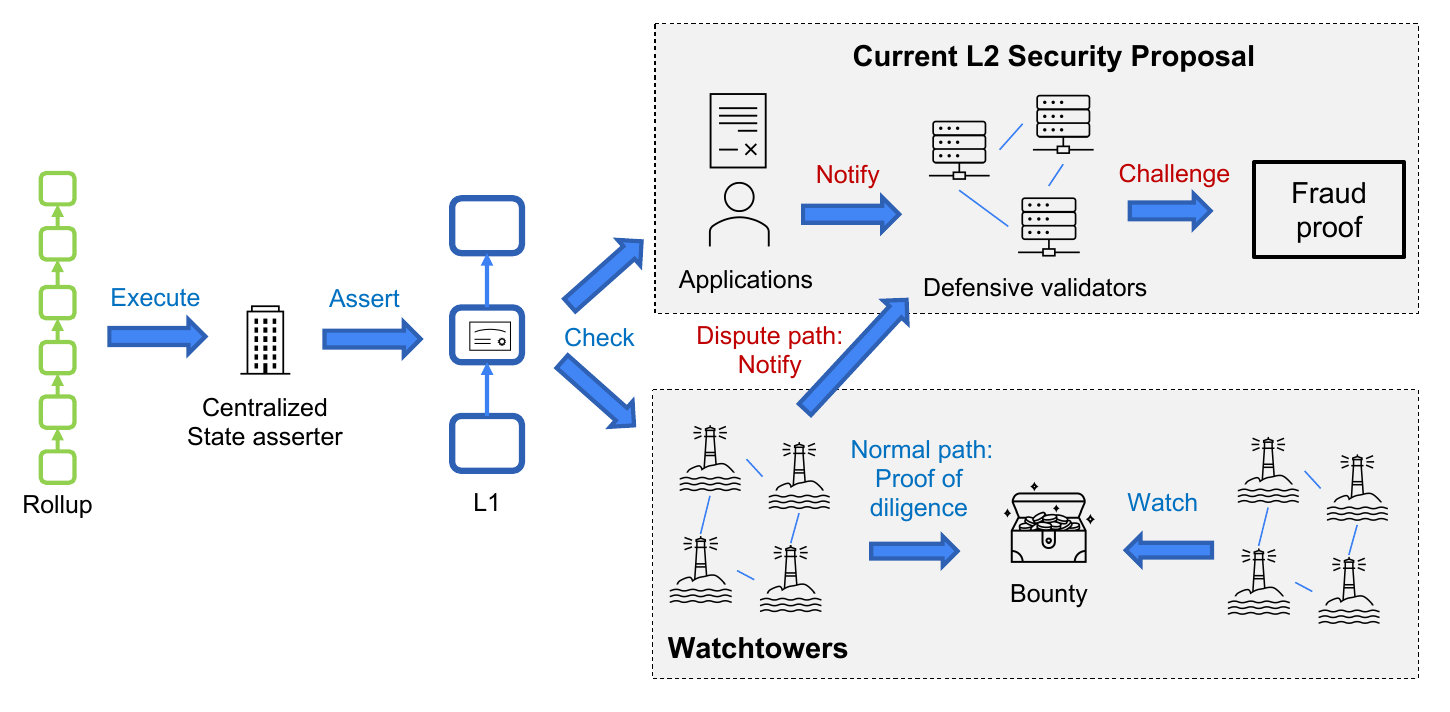}
    \caption{Watchtowers are added to the current L2 security workflow to guard normal path security.}
    \label{fig:system}
\end{figure*}

\paragraph*{Rollup security.} 
With the rapid adoption and reliance on rollups, there is an increasing need to address the critical security concern. Current rollup strategies secure the L2 states by requiring asserters to execute and commit L2 block data to the L1 blockchain (Figure~\ref{fig:system}).  These asserters, often staked and centralized, can be objectively slashed when their asserted states are proven to be incorrect. As depicted in the upper path of Figure~\ref{fig:system}, the system allows applications to independently verify the states committed on L1 and initiate disputes, serving as the secondary line of defense. These disputes are typically addressed and resolved via fraud proofs. 
However, a basic vulnerability remains despite the security layer provided by staked asserters and fraud proofs: the lack of incentives for actively watching the rollup. In the normal path, there is no guarantee that these applications monitor the asserted states consistently and effectively: how to ensure vigilance during normal path when attention might diminish due to the absence of apparent threats? In other words:  {\it who is watching the watchers? } This problem was identified sharply by the inventors of Aribtrum \cite{arbitrumblog}, however a systematic solution has remained open since then. 

\paragraph*{Watchtowers: the first line of defense for rollups.} In this paper we propose a ``rational watchtower pool'', a group of  workers {\em incentivized} to constantly watch the transactions \emph{in the normal path}. The lower part of Figure~\ref{fig:system} illustrates how  watchtowers operate independently, interacting with the existing rollup system only when an incorrect state is identified. At that point, they sound an alarm, much like what applications typically do. However, watchtowers are incentivized to stay vigilant at all times. Their role, serving as the first line of defense for rollups, is crucial for identifying potential faults that might otherwise go unnoticed. To ensure that the watchtower fulfill their ``watching responsibilities'' {\em diligently}, they must provide what we refer to as ``proof of diligence'', a prerequisite for earning incentives.

\paragraph*{Proof of diligence.}  Specifically, the duties of watchtowers entail that for new L2 state assertions to be integrated into the L2 ledger, watchtowers must execute the transactions and validate these new assertions. For a watchtower to demonstrate their diligence, their evidence must meet two criteria: (1) those who don't process the transactions should only be able to generate the proof with a negligible likelihood; (2) a proof produced by one watchtower shouldn't be valid for another. These standards ensure only the diligent watchtowers can generate valid proofs, and prevent multiple watchtowers from presenting identical proofs, thus promoting individual effort. To satisfy these conditions, each watchtower computes a verifiable random function (VRF) using the commitment of transaction execution trace. As every watchtower possesses a unique VRF key pair, they generate proofs independently and submit proofs on-chain, allowing public verification. As illustrated in Figure~\ref{fig:rollup-watchtower}, other watchtowers can recompute this proof, ensuring watchtowers presenting false proofs are identified and penalized during disputes. 

\paragraph*{Incentive framework for watchtowers.} For watchtowers to operate effectively, a carefully designed incentive mechanism is vital. This mechanism should consider two parts: positive incentives for continuous monitoring and negative incentives to punish undesirable actions. Positive rewards are allocated using a ``bounty mining'' scheme. Watchtowers calculate their proof of diligence to determine their eligibility for these rewards. To optimize efficiency, the criteria to obtain the bounty is set by a specific threshold, narrowing down the list of potential recipients. Only the selected winners need to submit their VRF outcomes as their proof of diligence. On the other hand, the negative incentives are implemented through staking. To engage in the protocol, watchtowers must deposit stakes, which are slashable upon detection of misconduct. Monitoring watchtower behavior introduces the problem of ``watching the watchtowers''. Hence, the incentive structure offers additional rewards for the verification of the proof of diligence. Watchtowers who spot inconsistent proofs can challenge them, earning a reward if their challenge is successful. We rigorously analyze the protocol as a non-cooperative game~\cite{maschler2020game}, examining the potential actions of watchtowers: being diligent or lazy. By appropriately configuring these incentives, our findings show that a strategy where all watchtowers diligently monitor rollups is the unique Nash equilibrium,  leading to an effective frontline defense for rollups.

\paragraph*{Extended mechanism design.} Moreover, we extend the action space for watchtowers to encompass potential collusion strategies. Here, several watchtowers might form a collusion to exchange execution results or decide on a mutual random outcome driven by the benefits of saving computational costs. Our findings indicate that when factoring in such cooperative actions, the equilibrium tends to favor collusion, nullifying the role of watchtowers. To counteract this, we introduce design enhancements to disrupt collusive behaviors. We design a whistleblower scheme allowing any colluder to secretly expose collusion in exchange for compensation. Our game-theoretic analysis demonstrates that the whistleblower system encourages betrayals against the collusion. Since this reporting remains confidential, while the act of betrayal is noticeable, the whistleblower can not be individually attributed. Consequently, the collusion will not be initiated at the first place.

\paragraph*{System design.} By integrating the proof of diligence and incentive mechanism, the design enforces watchtowers to maintain their essential roles in guarding the security of rollups in normal path. The network consists of a pool of watchtowers that are allocated to one of the participating rollups randomly, with the allocation changing over time. 
The system utilizes a staking mechanism to ensure sybil resistance and fair allocation of positive incentives; the stake is also used as a bond that can be utilized for negative incentives. On-chain contracts perform the broad activities of watchtower registration, rollup status monitoring, and disbursement of incentives. The off-chain client comprises a wrapper to the L2 full node that fetches intermediate information to mine the bounty. The system is implemented on the Optimism Bedrock stack~\cite{bedrock} that watches Optimism and Base rollups on the Goerli testnet. Note that since watchtowers operate independently of the optimistic rollup framework, the system is capable of monitoring multiple optimistic rollup chains simultaneously. We use Eigenlayer (EL)~\cite{eigenlayer} as the staking mechanism and build our registration functionality associated with it. Our implementation adds a very low compute cost to an L2 full node with a low transaction fee that can be adjusted on a sliding scale.

The subsequent sections provide a detailed breakdown of our investigative approach and findings. In Section \ref{sec:model}, we outline the system's model and explore various threat models. In Section~\ref{sec:background}, we compare prior work on related topics with our solution. Section \ref{sec:protocol} presents an in-depth view of our developed watchtower protocol. Further, in Section \ref{sec:security}, we extend the model to allow collusions and conduct a thorough cooperative game theoretic analysis of the security aspects and the incentive compatibility associated with the enhanced protocol. Section~\ref{sec:system} demonstrates our system design and implementation details, the evaluations results from the live system are presented in Section~\ref{sec:impl}. Concluding the paper, Section \ref{sec:discussion} engages in a comprehensive discussion, bringing to the fore the pivotal learning from our research and suggesting future directions. 
\section{Security Models and Definitions}
\label{sec:model}

\begin{figure*}
\centering
\begin{subfigure}{.45\textwidth}
  \centering
  \includegraphics[width=\linewidth]{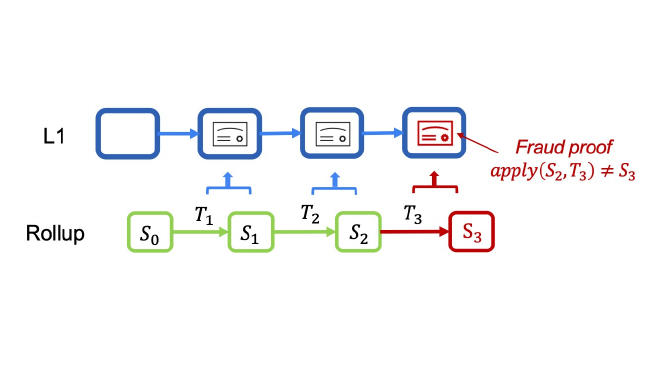}
  \caption{}
  \label{fig:rollup}
\end{subfigure}%
\begin{subfigure}{.45\textwidth}
  \centering
  \includegraphics[width=\linewidth]{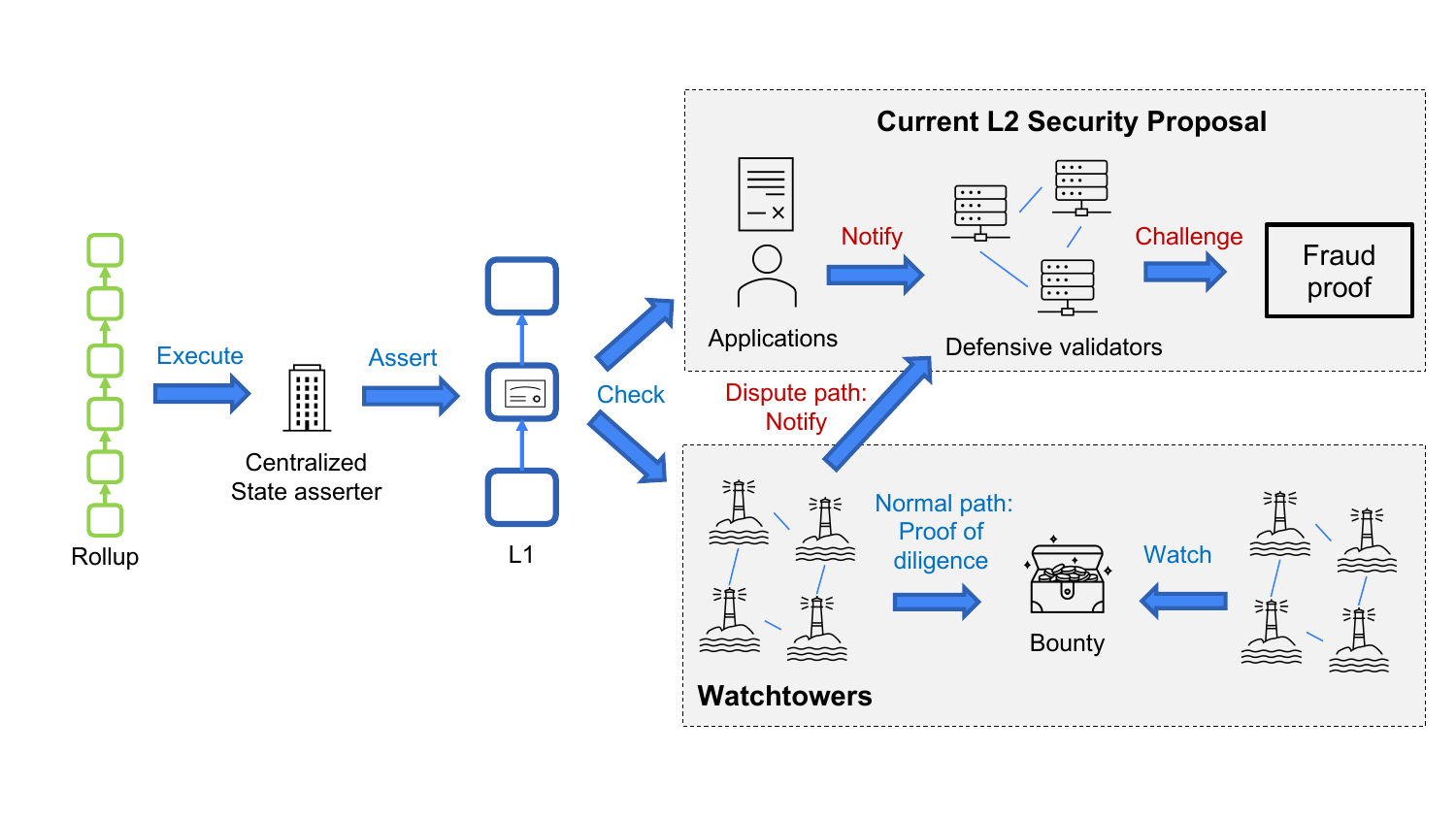}
  \caption{}
  \label{fig:watchtower}
\end{subfigure}
\caption{(a) Optimistic rollup model and (b) Watchtower model.}
\label{fig:rollup-watchtower}
\end{figure*}

\subsection{Rollup Model}
Rollups enhance the blockchain's efficiency by handling transaction execution off the main chain (L1). Present rollup techniques rely on either the validity proofs or fraud proofs to ensure security. Validity proofs, used in schemes known as zk-rollups, apply sophisticated cryptographic techniques to validate every batch of L2 transactions. On the other hand, fraud proofs, employed by optimistic rollups, come into play only when a fault is spotted in the execution process. Our focus is on enhancing the security of optimistic rollups (Figure~\ref{fig:rollup}). In this context, the watchtower pool has been introduced to monitor the normal-path operations -- that is, the process when transactions are presumed to be valid unless challenged. 

In our system, there exists an L2 chain employing optimistic rollup scheme. The scheme involves an untrusted \textit{asserter}, responsible for processing the L2 blocks and submitting the updated states assertion on L1 chain. L1 chain is guaranteed to be secure and live. We assume the data of L2 blocks are stored on L1 directly or through a trusted data availability service, hence the data is ensured to be retrievable. Furthermore, the system consists of a group of \textit{defensive validators}, whose role is activated when there is a need to produce a fraud proof, which is verifiable on chain. We consider a \textit{live} rollup system operator, who coordinates the issuance of L2 blocks (typically through a role called sequencer) and provides rewards for asserter, validators, and watchtowers. The goal of the operator is to ensure the security of rollup states with minimal cost.

\subsection{Watchtower Model}

Independent of the rollup infrastructure, our model incorporates a pool of $n$ registered entities known as \textit{watchtowers} (Figure~\ref{fig:watchtower}), denoted as $\{W_1, W_2, \cdots, W_n\}$. Each watchtower $W_i$ has deposited relative stake $\alpha_i$ into the system, where $\sum_{i}^{n}\alpha_i=1$ and the total amount of stake is $\mathcal{S}$. These watchtowers are modeled as individual \textit{rational adversary}, where the term ``rational'' implies that they operate with the purpose of optimizing a known payoff function. And they have the ability to methodically process all potential scenarios and select strategies that provides the best outcomes. 

Even though rational participants will not deviate from the protocols without cause, they might engage in attacks -- either intentionally or inadvertently -- if the expected rewards justify such actions. We identify several potential attacks in this framework, emphasizing the nuanced strategies rational actors might employ.


\paragraph*{Lazy watchtower} 
Since watchtowers are also required to execute the entire batch of L2 transactions, they can earn rewards by performing their duties diligently. This role closely resembles that of the asserter in the original optimistic rollup system, who posts computation results in exchange for rewards. As a result, one prominent challenge the watchtower design must confront is the ``lazy watchtower'' problem. This issue arises because of two main reasons: (1) rational watchtowers may submit arbitrary responses if the results lack verification process, and (2) they might opt out of protocol participation if the associated costs outweigh potential rewards. In essence, the watchtowers must provide a form of evidence for their work and the protocol must offer sufficient incentives to encourage participants to actively and consistently perform tasks.

\paragraph*{Collusion attack.} Rational entities might form collusion where several parties conduct coordinated actions based on a common agreement. However, it's essential to note that, despite the existence of any agreement, colluding parties maintain individual autonomy and continue to prioritize their self-interest. Within the scope of collusion attacks, we consider adaptive adversaries who can adjust their collusion strategies in response to protocol developments.

In our system, we assume that adversaries are limited by computational constraints, so that they are not able to break the security of necessary cryptographic primitives. It's important to highlight that rational adversaries, guided by their payoff functions, are considered weaker threats compared to Byzantine adversaries. Since their profit-driven actions exclude certain strategies. In the discussion section (Section~\ref{sec:discussion}) we explore the trade-off between different levels of security and associated costs. 


\begin{table}[h!]
\centering
\caption{Summary of notations.}
\label{tab:summary_terms}
\begin{tabular}{ll|ll}
\toprule
\textbf{Term} & \textbf{Definition} & \textbf{Term} & \textbf{Definition} \\
\midrule
$n$ & Number of registered watchtowers & $c_T$ & Cost per transaction batch \\
$W_i$ & Watchtower $i$ & $c_V$ & Cost of resolving disputes \\
$\alpha_i$ & Relative stake of watchtower $W_i$ & $R_B$ & Normal path reward for watchtowers \\
$\mathcal{S}$ & Total amount of stake in the system & $R_C$ & Dispute path reward for watchtowers \\
$S$ & Blockchain state & $R_w$ & Whistleblower protocol reward \\
$T$ & A batch of L2 transactions & $\phi(\alpha_i)$ & Bounty mining threshold \\
$E$ & Execution trace & $t$ & Deposit for participating in collusion \\
$r_S$ & State assertion & $h$ & Rent in diligent collusion \\
$r_E$ & Merkle root of execution trace & $n_c$ & The size of the colluding group \\
\bottomrule
\end{tabular}
\end{table}

\subsection{Preliminaries}
We introduce fundamental primitives utilized throughout the paper. Other notations are summarized in Table~\ref{tab:summary_terms}. 

\paragraph*{Verifiable random functions.} Our protocol use verifiable random functions (VRFs), providing two functions to generate and verify proofs. $\hash(x)$ processes an input $x$ and returns two values $(d, \pi_d)$: a normalized hash digest $d$ and a proof $\pi_d$. The value $d \in [0,1)$ is uniquely determined by the input $x$ and a secret key $sk$, and is indistinguishable from a random value to anyone that does not know $sk$. $\verifyVRF(\pi_d, d, x) \in \{\true, \false\}$ takes the proof $\pi_d$ input and allows anyone who knows the public key $pk$ to verify whether $d$ is the correct value computed from $x$ and $sk$.

\paragraph*{Merkle trees.} Merkle trees are a fundamental data structure in cryptography, summarizing lists of items, such as transactions or states, by concatenating their cryptographic hashes at various levels of the tree. We provide the function \Merklize$(L) \to r$ that generates a Merkle root $r$ from a list of items $L$. It involves the construction of a Merkle tree or Patricia trees \cite{merkle1987digital, wood2014ethereum} for $L$ and then produces the root hash that represents the entire list of items. \MerkleProof$(r, l, L)\to p$ is used to generate a Merkle proof $p$ associated with a leaf node $l$ in a tree rooted at $r$. This proof consists of the minimal amount of information needed to confirm the presence of the specific leaf node $l$ within the tree constructed from $L$. Furthermore, there exists a validation function \verifyMerkleProof$(r, l, p) \in \{\true, \false\}$ that takes in a root $r$, a leaf $l$, and a Merkle proof $p$, then returns a boolean value indicating whether $p$ demonstrates that $l$ is indeed part of the Merkle tree rooted at $r$.

\paragraph*{State transitions.} We consider a general model for L1 blockchain, which keeps the latest states set $S$ and transactions organized as blocks within hash-based chains. The function \apply$(S', T) \to (S, E)$ represents the process of applying a list of transactions $T$ to a prior state $S'$ to yield a new state $S$ and an execution trace $E$. The execution trace is a detailed record of all intermediate states during the transition from $S'$ to $S$, serving as a reference for verification. Another function \validate$(r,r',L)\in \{\false(r), \false(r'),\false(r,r')\}$ resolves a conflict between two Merkle roots $r$ and $r'$ calculated from the same list of items $L$. The output represents the subset of roots that are proven to be invalid. There are different ways to implement this function, such as using L1 as a trusted third party to provide ground truth or executing an interactive verification game (IVG) to identify the incorrect root.

\paragraph*{Game theory.} The concepts about game theory utilized in our analysis are all defined in the book~\cite{maschler2020game}. We first analyze the rollup security with watchtowers as a non-cooperative game in strategic form (Def.  4.2, \cite{maschler2020game}), and examine the dominance of strategies (Def.  3.10 and 4.6, \cite{maschler2020game}) to find the pure strategy Nash equilibrium (Def. 4.17 and Def. 5.3, \cite{maschler2020game}) for watchtowers. We also discuss an enhanced protocol considering cooperative game (Chapter 15, \cite{maschler2020game}), where more than one Nash equilibria exist and Pareto efficiency (Def. 15.7, \cite{maschler2020game}) is considered. 

\section{Background and Prior Work}
\label{sec:background}
The dialogue surrounding optimistic rollups originated in community discussions on the Ethereum research forum \cite{rollupdiscussion}, where developers and researchers shared insights about potential scalability solutions. One of the earliest detailed presentations came from the Optimism team, who outlined the foundational framework of optimistic rollups and their theoretical implications~\cite{optimismrollup}. The theoretical cornerstone for optimistic rollups was further strengthened through scholarly research. Pioneering works such as Arbitrum~\cite{kalodner2018arbitrum} and TrueBit~\cite{truebit} introduced essential concepts related to off-chain computation and dispute resolution.

Subsequent discussions began to address broader implications in contexts such as data availability~\cite{al2018fraud,yu2020coded,nazirkhanova2022information}, sequencer risks~\cite{mccorry2021sok,motepalli2023sok}, and validator behavior~\cite{koch2018predictable,arbitrumblog,nabi2020game,proofofcustody,dong2017betrayal}. Specifically, among these topics, validator behavior is most pertinent to this paper. While the basic rollup design discusses an elementary incentive structure to motivate asserters to post correct states, verifiers may lack incentives to monitor the system's states diligently, such a situation is known as the verifier's dilemma problem~\cite{luu2015demystifying}. To address this issue and enhance system security, the authors of Arbitrum~\cite{arbitrumblog} proposed an attention game in which verifiers who fail to participate may be punished. TrueBit~\cite{koch2018predictable} also suggested an enhanced mechanism to select a pool of validators, requiring them to submit a proof of independent work~\cite{proofofcustody} to enforce verification. Additionally, some research \cite{nabi2020game} provides game-theoretic analyses of collusion risks in incentivized computation outsourcing, proposing mitigation strategies, though not completely resolving the issue. In the field of verifiable outsourced cloud computing, \cite{dong2017betrayal} investigates the dynamics between clients and workers within smart contract frameworks, demonstrating a preference for honest behavior under certain conditions. Although they propose a ``traitor contract'' to eliminate collusion, this is limited to a two-party context. Furthermore, we contend that this proposed solution might not be effective, as the contract's output could compromise the traitor's secrecy. Another study \cite{okada1993possibility} examines cooperation in $N$-person prisoner's dilemma scenarios, with institutional arrangements akin to smart contracts. Differing from our focus, this work considers collusion through bargaining rather than a leader-based approach. A new design of the attention game was proposed in a recent paper~\cite{mamageishvili2023incentive} to find the optimal number of validators that minimizes failure probability. The design only provides probabilistic security and requires modification in the underlying rollup protocol to ensure deterministic security. As a comparison, we provide a plug-and-play solution that can be used for any off-chain compute resource with minor modifications.


\section{The Watchtower Protocol}
\label{sec:protocol}

\subsection{Proof of Diligence}

In our protocol, we focus on a single task that watchtowers undertake: verifying the updated state assertion $r_S$ of the L2 blockchain, a state assertion is the Merkle root of the state tree, calculated by $r_S = \Merklize(S)$. The states $S$ on a blockchain consist of a list of key-value pairs, such as the account address and the account balance. Formally, given the latest validated state $S'$ and a sequence of transactions $T$, the responsibility of the watchtowers is to verify $r_S$ by executing the computations specified in Algorithm~\ref{alg:checkstate}.

\begin{algorithm}[t]
\caption{Function for Watchtowers}
\label{alg:checkstate}
\begin{algorithmic}[1]
\Function{\checkState}{$S',T,r_S,\alpha_i$} 
\State $(S,E) \gets \apply(S',T)$
\State $r'_S = \Merklize(S)$
\State $r_E = \Merklize(E)$ \label{alg:podfirst}
\State $(d, \pi) = \hash(r'_S | r_E)$
\If{$d < \phi(\alpha_i)$}
    \State submit proof of diligence $(d, \pi)$
\EndIf\label{alg:podlast}
\If{$r_S = r'_S$} 
    \State Return \true
\Else
    \State Return \false
\EndIf
\EndFunction
\end{algorithmic}
\end{algorithm}

Algorithm~\ref{alg:checkstate} represents the process of assessing whether the proposed state assertion $r_S$ is indeed consistent with the ledger history. The watchtower calling the $\checkState$ function first applies the transactions $T$ to the initial state $S'$, which returns the new state $S$ and an execution trace $E$. The watchtower then calculates the Merkle root of the legitimate state $S$ and compares the result $r'_S$ with the posted $r_S$. If the verification process succeeds and no faults are found, the watchtower considers the new state $S$ as validated. Otherwise, the watchtower is expected to raise an alarm to the rollup scheme, activating defensive validators for dispute resolution and fraud proof creation.

Besides, the watchtowers utilize the execution trace $E$ to construct a Merkle root $r_E$. Since $r_E$ is not available anywhere and can only be derived by executing the transactions, it serves as evidence of their work and diligence. They compute a VRF using $r_E$ and $r'_S$, incorporating their secret keys. It is assumed that the corresponding public keys were disclosed during the registration phase. The VRF produces $(d, \pi)$; the digest $d$ is subsequently used to allocate rewards for diligent watch. Accompanied by the proof $\pi$, the watchtower submits this as the \textit{proof of diligence}, which satisfies the following properteis:
\begin{itemize}
    \item[-] \textbf{Verifiability.} Given a keypair $(pk, sk)$, for any input $x$, if $(d, \pi) \gets \hash(x)$, then $\verifyVRF(\pi, d, x) = \true$.
    \item[-] \textbf{Uniqueness.} Given a keypair $(pk, sk)$, for any input $x$, if $(d, \pi) \gets \hash(x)$, no one, including the key owner, can produce a different $d'\neq d$ and the associated proof $\pi'$ such that $\verifyVRF(\pi', d', x) = \true$.
    \item[-] \textbf{Pseudorandomness.}  For a given input $x$ and public key $pk$ the output $d$ is indistinguishable from a truly random string to anyone who does not possess the private key $sk$.
\end{itemize}

The verifiability ensures that once the proof is posted on the chain, it can be verified by every watchtower using $r_S, r_E$, and the public key $pk$. This process is referred to as  ``watching the watchtowers''. Other watchtowers will use their own $r_S, r_E$ values to check the proof. If they detect any inconsistencies, they invoke the $\validate$ function to resolve conflicts on the chain. After a predefined challenge period $t_C$, the protocol concludes that the remaining proofs are correct. The uniqueness and pseudorandomness imply that only diligent watchtowers can generate a valid proof, and the proof generated by one watchtower cannot be used by others. 

\subsection{Incentive Design}
To ensure that watchtowers execute their verification duties with diligence, it's imperative to institute a carefully designed incentive mechanism. First of all, this mechanism mandates that all enrolled watchtowers use their stakes as the deposit. This deposit acts as a form of commitment to honest service — any verified misbehavior results in the slashing of their deposit. 

\paragraph*{Bounty mining.} Understanding that watching operations incur costs, represented as $c_T$, in the process of executing all transactions in $T$, we introduce a bounty mining scheme to motivate the watchtowers to perform the verification (specified in Algorithm~\ref{alg:checkstate}, line~\ref{alg:podfirst}-\ref{alg:podlast}). The process of bounty mining can be analogous to the process of committee or leader selection in some blockchains~\cite{chen2016algorand,david2018ouroboros}, where a VRF is computed to determine bounty winners. Formally, a diligent watchtower $W_i$ with relative stake $\alpha_i$ who performs the transaction execution can generate a proof of diligence $(d,\pi)$ on the task. 
Then with a probability $\phi(\alpha_i)$, $W_i$ finds that its proof satisfies a certain condition (specified in Eq.~\ref{eq:cond}), allowing them to receive a bounty by publishing the proof. A system parameter $\theta$ is defined to control the probability that a party with all stake wins the bounty.
\begin{equation}
    \Pr[\vrf{i}(r_S| r_E).d < \phi(\alpha_i)] = \phi(\alpha_i) = 1-(1-\theta)^{\alpha_i}
    \label{eq:cond}
\end{equation}

The amount of bounty each winner who submits a valid proof $(d,\pi)$ can collect is a constant value $R_B$.  If any watchtower identifies an incorrect proof, the $\validate$ function will be invoked to resolve the dispute. In this process, both watchtowers are required to publish their $r_S, r_E$ values. The winner of the challenge will receive a constant reward of $R_C$, and a compensation distributed among all winners for the cost of dispute resolution $c_V$. The rewards setting satisfies the following conditions:
\begin{align}
    \label{eq:cond1} R_B > \frac{c_T}{\phi(\alpha_0)}, \, R_C > c_T
\end{align}
where $\alpha_0$ is the unit stake fraction, hence $\alpha_0\le \min\{\alpha_i\}_{i\in[1,n]}$. To ensure that the slashed deposit is enough to pay for the cost of dispute resolution and rewards, we require that 
\begin{equation}
\label{eq:cond2}
    \alpha_0S \ge c_V + (n-1)R_C
\end{equation}
In summary, the entire protocol works as follows.
\begin{enumerate}
    \item When a new state assertion $r_S$ is published, a bounty timer $t_1$ starts. Each watchtower recomputes the assertion $r'_S$ and generates the execution trace root $r_E$ by applying transactions $T$ to the old states $S'$, then computes $(d,\pi)=\hash(r'_S, r_E)$. 
    \item If the assertion is incorrect ($r_S \neq r'_S$), the watchtower notifies the defensive validators of the rollup to initiate a challenge.
    \item If a watchtower wins the bounty, the watchtower submits proof of diligence $(d,\pi)$ before $t_1$.
    \item When a watchtower observes a proof of diligence submitted by other watchtowers, it verifies the proof using the other's public key and the execution trace root $r_E$ calculated by itself. If the proof is incorrect, the watchtower calls $\validate$ interface to resolve the dispute.
    The cost of triggering $\validate$ is denoted as $c_V$ shared among all watchtowers who call the function.
    \item If no $\validate$ is triggered before $t_1$ expires, the rollup operateor concludes that the asserted execution trace root is correct. Validated bounty winners receive $R_B$ as reward each. 
    \item If $\validate$ is triggered, the winning parties receive a reward $R_C$ and a compensation for the shared cost, and the losing parties lose all the stake. 
\end{enumerate}




\subsection{Incentive Analysis in Non-Cooperative Games}

We first consider the proof of diligence protocol as a non-cooperative game denoted as \podgame, where different watchtowers optimize their individual payoff without any mutual agreement for cooperation. Watchtowers can adopt one of two possible strategies: diligent or lazy. Diligent watchtowers execute transactions honestly and report proof of diligence when the condition is met. In contrast, lazy watchtowers opt for generating a random result as the new state assertion, incurring negligible cost. Moreover, these lazy watchtowers might submit a fake proof, computed from the random root, to deceitfully claim the bounty. Notably, we consider only those lazy {\em and deceitful} watchtowers, as non-deceitful lazy watchtowers are indistinguishable from non-participants in impact. A default constraint of our incentive mechanism is that the payoff for diligent behavior is always positive, which dominates the non-participating strategy. Therefore, we omit this trivial case in the subsequent analysis.

Let $u^a_i(n_d)$ be the expected payoff function of watchtower $W_i$ given that watchtower $W_i$ choose action $a\in\{d,l\}$,  where $d$ and $l$ represent diligent and lazy strategies respectively, and there are $n_d$ diligent watchtowers in total. According to the protocol, for all $i\in[1,n]$, these payoff functions are defined as follows:
\begin{align}
    u^d_i(n_d) &= \begin{cases}
        \phi(\alpha_i)R_B+R_C-c_T & n_d < n\\
        \phi(\alpha_i)R_B-c_T & n_d = n
    \end{cases}\\
    u^l_i(n_d) &= \begin{cases}
        \phi(\alpha_i)R_B & n_d = 0\\
        -\alpha_iS\phi(\alpha_i) & n_d > 0
    \end{cases}
\end{align}
Note that a lazy watchtower would attempt to mimic genuine probability to submit proofs, otherwise the proof submission frequencies that can happen with negligible probability can be used to detect malicious behaviors. By comparing the payoff for different strategies, we observe that the diligent strategy dominates the lazy strategy for all watchtowers; formally, we have the following theorem.
\begin{theorem}
\label{thm:pod-dominant}
    The diligent strategy in the \podgame\ is a dominant strategy for all watchtowers.
\end{theorem}

\begin{proof}
    For every watchtower $W_i$, given an arbitrary strategy vector $s_{-i}$ containing all others' actions. Let $n_d$ denote the number of watchtowers that are diligent in $s_{-i}$, we compare the payoff of two strategies for $W_i$ below.
    \begin{itemize}
        \item[-] Case 1: If $n_d = 0$, due to Eq.\ref{eq:cond1},  $u^d_i(1) =\phi(\alpha_i)R_B+R_C-c_T >  \phi(\alpha_i)R_B =  u^l_i(0) $.
        \item[-] Case 2: If $0<n_d < n-1$, due to Eq.\ref{eq:cond1},  $u^d_i(n_d+1) =\phi(\alpha_i)R_B+R_C-c_T >  -\alpha_iS\phi(\alpha_i) =  u^l_i(n_d) $.
        \item[-] Case 3: If $n_d = n-1$, due to Eq.\ref{eq:cond1},  $u^d_i(n) =\phi(\alpha_i)R_B >  -\alpha_iS\phi(\alpha_i) =  u^l_i(n-1) $.
    \end{itemize}
\end{proof}
Theorem~\ref{thm:pod-dominant} implies that the game has a unique Nash equilibrium since we can eliminate the strictly dominated lazy strategy and get a unique strategy vector of all diligence; this follows directly from Cor.4.37 \cite{maschler2020game}.

\begin{corollary}
    The unique Nash equilibrium in \podgame occurs when every watchtower is diligent.
\end{corollary}

In conclusion, in a setting without cooperation, our proof of diligence protocol ensures that rational watchtowers will always work diligently, providing the first line of defense for the rollup system. In practice, this setting can model many situations, such as when watchtowers cannot communicate with each other, or when there is a public reputation system where participating in any cooperation would be detected and deteriorate reputations.
\section{Cooperative Games and The Enhanced Protocol}
\label{sec:security}

In the \podgame, we observe that if all watchtowers choose to be lazy and agree to use a common random $r_E$ and $r_S$ to compute their proofs, they will receive a higher payoff than in the all-diligent equilibrium. However, this strategy did not get chosen since any party can deviate from this ``unreliable collusion'' to achieve higher utility, while lazy parties end up losing all their stakes. 

Collusion can be reinforced by adding specific punishment mechanisms. In the context of rollups, smart contracts are the most viable tools for enforcing agreements or promises of such cooperation. Beyond the previously mentioned lazy collusion strategy, other strategies may improve payoff through cooperation, like sharing the costs of diligence. Additionally, the process of forming a colluding group can vary. For instance, in a leader-based method, a watchtower might take the initiative to create a colluder contract \cite{dong2017betrayal}, setting conditions for joining and outlining actions to be taken, thereby allowing others to join. Conversely, in a leader-less method, watchtowers can choose to join a collusion group and negotiate group strategy collectively \cite{okada1993possibility}.

In this section, we explore the space of collusion, concentrating on two main leader-based strategies applicable to most relevant settings. Our primary findings reveal that establishing mutual agreements enforced by smart contracts makes lazy actions more beneficial. To counteract this, we propose setting up Whistleblower contracts, which encourage colluders to betray their collusion, thereby eliminating the lazy equilibrium. 

\subsection{Lazy Collusion}
\label{sec:lazy-collusion}

One possible strategy that watchtowers might employ to solidify the collusion is to require all colluders to deposit a certain amount of stakes into the collusion. If a colluder posts a proof computed from different execution roots, they will lose the collusion deposit.

We assume any party is capable of initiating such collusion, and we refer to that party as the leader. A leader will specify the amount of stake $t$ that each newly joined colluder needs to contribute. Since the collusion is motivated by the benefit of being lazy, we term this strategy ``lazy collusion.'' We observe that watchtowers choosing not to join the collusion perceive the same game as \podgame, and therefore, they are likely to adopt a diligent strategy. If such independent watchtowers exist, lazy collusion will not gain the expected advantage by not computing the results. Consequently, collusion will only be effective when all watchtowers participate in it. Specifically, the process of forming lazy collusion unfolds as follows:
\begin{enumerate}
    \item A watchtower initiates collusion by placing a deposit of $t$. The watchtower also releases a randomly chosen $r'_E$.
    \item Other watchtowers may join the collusion by placing a deposit of $t$. If $n$ watchtowers join the collusion before $t_{lc}$, the collusion is formed. Otherwise, watchtowers get back their deposits.
    \item During the watching phase, all colluders are required to calculate the proof of diligence using $\hash(r_S, r'_E)$, where $r_S$ is the state root posted by the asserter.
    \item If a colluder becomes a winner, the collusion protocol will check whether the winner's proof is calculated from $r'_E$, if not, the winner is considered a traitor and will lose $t$. 
    \item At the end of the collusion, colluders who do not betray the collusion receive $t + n_tt/(n_c-n_t)$, where $n_c$ is the size of colluding group, $n_t$ is the number of traitors. If all colluders betray, everyone gets back their deposit $t$.
\end{enumerate}

We are considering two possible actions that all colluders (including the leader) can take, given the collusion strategy selected by the leader: obey and betray. Colluders who choose to obey the strategy follow the leader's instructions to submit a response calculated from the specified $r_E$, while those who choose to betray may submit something different, driven by personal interest. The game induced by lazy collusion is denoted as $\lcgame$. Let $u_{l_i}^a(n_o)$ be the expected payoff function of the $i$-th colluder $W_{l_i} (i \in [1,n_c])$. $a \in \{o,b\}$ represents for the action chosen by  $W_{l_i} $, with $o$ as obey and $b$ as betray. There are $n_o$ colluders who choose to obey the collusion strategy. According to the protocol, $n_c = n$, so we let $l_i = i$ for simplicity, and the payoff functions can be written as follows:
\begin{align}
    u^o_i(n_o) &= \begin{cases}
        -\alpha_i\mathcal{S}\phi(\alpha_i) + \frac{(n-n_o)t}{n_o} & n_o < n\\
        \phi(\alpha_i)R_B & n_o = n
    \end{cases}\\
    u^b_i(n_o) &= \begin{cases}
        \phi(\alpha_i)R_B - c_T & n_o = 0\\
        \phi(\alpha_i)R_B + R_C - c_T -t & n_o > 0
    \end{cases}
\end{align}

Now when there exists such lazy collusion, we compare the payoff of colluders with different actions and observe that obeying the group strategy dominates the betrayal for all colluders when the deposit $t$ is high enough, formally, we have the following theorem.
\begin{theorem}
\label{thm:lc-dominant}
    The ``obey'' strategy in the \lcgame\ is a dominate strategy for colluder $W_i$ if the following conditions hold:
    \begin{align}
     \label{eq:t1}   t &> R_C-c_T\\
      \label{eq:t2}  t &> \frac{n-1}{n}\left(\alpha_i\mathcal{S}\phi(\alpha_i) + \phi(\alpha_i)R_B + R_C-c_T\right)
    \end{align}
\end{theorem}

\begin{proof}
    For every colluder $W_i$, given an arbitrary strategy vector $s_{-i}$ containing all others' actions. Let $n_o$ denote the number of colluders that obey the collusion in $s_{-i}$, we compare the payoff of two strategies for $W_i$ below.
    \begin{itemize}
        \item[-] Case 1: If $n_o = n-1$, due to Eq.~\ref{eq:t1},  $u^o_i(n) =\phi(\alpha_i)R_B >  \phi(\alpha_i)R_B + R_C-c_T-t =  u^b_i(1) $.
        \item[-] Case 2: If $0<n_o < n-1$, due to Eq.~\ref{eq:t2} ,  $u^o_i(n_o+1) =-\alpha_i\mathcal{S}\phi(\alpha_i) + \frac{(n-n_o-1)t}{n_o+1} > \phi(\alpha_i)R_B + R_C-c_T-t  =  u^b_i(n_o) $.
        \item[-] Case 3: If $n_o = 0$, due to Eq.~\ref{eq:t2},  $u^o_i(1) =-\alpha_i\mathcal{S}\phi(\alpha_i) + (n-1)t > \phi(\alpha_i)R_B-c_T =  u^b_i(n) $.
    \end{itemize}
\end{proof}

Then we consider the game combining \podgame\ and \lcgame. It starts with an initiation phase where watchtowers may choose to initiate a contract to become the collusion leader. If all watchtowers join the same collusion contract, which is the only case where collusion will be formed successfully, and the conditions specified in Eq.\ref{eq:t1} and \ref{eq:t2} are satisfied, this sub-game \lcgame\ can be eliminated according to Theorem~\ref{thm:lc-dominant}, with the payoff induced by the dominant strategy. Otherwise, independent watchtowers will follow the \podgame\ and iterative elimination can be applied with Theorem~\ref{thm:pod-dominant}. Observing that the payoff derived from \lcgame\ is higher than \podgame,  we find the game has two Nash equilibria but only the lazy collusion strategy is Pareto efficient.

\begin{corollary}
    In \podgame\ that allows lazy collusion, there are two Nash equilibria: (1) all watchtowers are independently diligent (2) all watchtowers are collusively lazy. The second equilibrium is Pareto efficient.
\end{corollary}

\subsection{Diligent Collusion}
Moreover, watchtowers might choose to remain diligent while seeking to form a collusion to share the execution costs. In this scenario, the leader initiating the collusion carries out the computations and commits its solution to the group. Anyone wishing to join the collusion is required to contribute a fee of $h < c_T$ to access the results. Subsequently, all colluders utilize the execution root, as calculated by the collusion leader, to generate proof of diligence and claim bounties. We refer to this strategy as ``diligent collusion''. The process for forming such a diligent collusion unfolds as follows:

\begin{enumerate}
    \item A watchtower initiates collusion by placing a deposit of $t$. The watchtower also commits a computed $r'_E$ and specifies a rent $h < c_T$.
    \item Other watchtowers may join the collusion by paying the rent of $h$. Then the committed $r_E$ is revealed. 
    \item During the watching phase, all colluders are required to calculate the proof of diligence using $\hash(r_S, r'_E)$.
    \item If the proof a non-leader colluder submits is recognized as a faulty proof, the leader will lose $t$, and others will get back $ t / (n_c-1)$, where $n_c$ is the size of the colluding group. 
    \item If the proof provided by leader gets accepted, the leader gets $t + (n_c-1)h$.
\end{enumerate}

\begin{table}[]
\centering
\caption{\dcgame: The game induced by diligent collusion.}
\label{tab:dc-game}
\begin{tabular}{|cc|cc|}
\hline
\multicolumn{2}{|c|}{\multirow{2}{*}{Payoff \payoff{u_{d_1}\\u_{d_i}}}} & \multicolumn{2}{c|}{follower}                \\ \cline{3-4} 
\multicolumn{2}{|c|}{}                                                    & \multicolumn{1}{c|}{All join} & Not all join \\ \hline
\multicolumn{1}{|c|}{\multirow{3}{*}{leader}} &
  Obey &
  \multicolumn{2}{c|}{\payoff{\phi(\alpha_{d_1})R_B-c_T+(n_C-1)h\\ \phi(\alpha_{d_i})R_B-h}} \\ \cline{2-4} 
\multicolumn{1}{|c|}{} &
  Betray &
  \multicolumn{2}{c|}{\payoff{\phi(\alpha_{d_1})R_B-c_T+R_C-\frac{c_V}{n-n_C+1}-t\\-\alpha_{d_i}\mathcal{S}\phi(\alpha_{d_i})+\frac{t}{n_C-1}}} \\ \cline{2-4} 
\multicolumn{1}{|c|}{} &
  Cheat &
  \multicolumn{1}{c|}{\payoff{\phi(\alpha_{d_1})R_B+(n-1)h\\ \phi(\alpha_{d_i})R_B-h}} &
  \payoff{\phi(\alpha_{d_1})R_B+(n-1)h\\ \phi(\alpha_{d_i})R_B-h} \\ \hline
\end{tabular}
\end{table}

In the game of diligent collusion, we observe that the leader has three possible actions: obey, betray, and cheat. ``Obey'' implies that the leader will diligently compute the transaction execution root and share it with others. ``Betray'' suggests that the leader might commit a random output to the colluding group while submitting a correct proof for its own benefit. And ``cheat'' represents the scenario in which the leader lazily commits and submits the same random output. The other colluders may only choose to follow what the leader commits, since they pay the rent; in other words, there is no benefit to join the collusion if they plan to choose another action. Table~\ref{tab:dc-game} lists the payoff functions of $\dcgame$, the game induced by the diligent collusion strategy. $u_{d_1}$ and $u_{d_i} (i\in[2,n_c])$ represent for the payoff function for the leader and other colluders. Note that if the leader chooses to cheat, its payoff is highly influenced by whether all watchtowers join the collusion. If there exist independent watchtowers, they must choose the diligent strategy as the $\podgame$ implies, then all colluders will be punished by the proof of diligence protocol. Therefore, it is evident that the following theorem holds:

\begin{theorem}
\label{thm:dc}
$\dcgame$ has no pure strategy Nash equilibrium when $ t> R_C - c_V/(n-1)-h$.
\end{theorem}

\begin{proof}
    We denote the strategy profile in $\dcgame$ as $\{a, n_c\}$, where $a\in\{o,b,c\}$ is the action chosen by leader, representing obey, betray and cheat, $n_c$ is the number of other watchtowers who choose to join the collusion. Firstly, the condition $t> R_C - c_V/(n-1)-h$ implies that $\forall n_c\in[2,n], u_{d_1}(\{o, n_c\}) > u_{d_1}(\{b, n_c\})$, hence betray is strictly dominated by obey. We then observe that $u_{d_1}(\{o, n\}) < u_{d_1}(\{c, n\})$, indicating that $\{o, n\}$ is not a Nash equilibrium, as in this case the leader can achieve a higher payoff by switching to cheat. Additionally, independent watchtowers receive a higher payoff if they join the collusion when the leader chooses to obey, as this spreads the execution cost across the entire colluding group. Conversely, when the leader opts to cheat, if all watchtowers join the collusion, switching to be independently diligent achieves a better payoff. However, if not all watchtowers join, obey becomes the more beneficial strategy for the leader. Consequently, there doesn't exist any pure strategy that is a Nash equilibrium.
\end{proof}

Theorem~\ref{thm:dc} indicates that even if we take \dcgame\ into consideration, the pure strategy Nash equilibria of the full game remain the same. Hence we have the following property.

\begin{corollary}
    In \podgame\ that allows lazy and diligent collusion, there are two Nash equilibria: (1) all watchtowers are independently diligent (2) all watchtowers are collusively lazy. The second equilibrium is Pareto efficient.
\end{corollary}

\subsection{Enhanced Protocol with Whistleblower}
The incentive for colluders to establish collusion lies in the potential to mine bounties with less effort. However, less effort always correlates with the likelihood of faulty proofs. As analyzed in Section~\ref{sec:lazy-collusion}, watchtowers are incentiviced to join the lazy collusion. In lazy collusion, the benefits derived from betraying the collusion (by being diligent in computations) do not suffice to offset the loss of the collusion deposit. Thus, the principles for eliminating the collusion are to: (1) offer rewards for identifying and reporting collusion, and (2) provide compensation for the losses incurred from betraying the collusion. We refer to such colluders, who disclose information about collusion to the rollup system, as ``whistleblowers''. In response, we have specifically designed the whistleblower protocol as follows:
\begin{enumerate}
    \item The rollup operator establishes a whistleblower bounty $R_w$ and declares that the first whistleblower will be eligible for the reward.
    \item Any individual can place a deposit of $d$ and submit the correct $r_E$ to assume the role of a whistleblower.
    \item The protocol invokes the $\validate$ interface to resolve the dispute. If the whistleblower succeeds, they receive $R_w + d$ in return. Otherwise, a loss results in the forfeiture of their deposit.
\end{enumerate}

To ensure the payoff of whistleblower is better in all above collusion games, we derive the following condition to determine rewards.

\begin{lemma}
    With the additional action ``report'' that each watchtower can choose, the strategy that all watchtowers obey the lazy collusion is no longer a Nash equilibrium for watchtower $i$ if
    \begin{align}
    \label{eq:whistleblower}
    R_w &> \phi(\alpha_i)R_B + c_V + \alpha_i\mathcal{S}\phi(\alpha_i) + c_T
\end{align}
\end{lemma}

\begin{proof}
We first consider the impact that the additional action \textit{report} brings to \lcgame. First, the whistleblower, by adhering to the collusion agreement to submit proof, will not be subject to punishment by the slashing rule of the collusion. 
However, the payoffs of other colluders will be reduced due to the exposure of the collusion, hence the changes in outcomes are detectable and might be used to augment the collusion deposit. Even though, the act of reporting cannot be traced back to an individual colluder.  Therefore, any colluder may switch to report to gain higher payoff from the whistleblower protocol, since
\begin{align}
    u_i^r(n) &= -\alpha_i\mathcal{S}\phi(\alpha_i) - c_V + R_C - c_T + R_w \\
    &> \phi(\alpha_i)R_B + R_C >  \phi(\alpha_i)R_B = u_i^o(n)
\end{align}
This further discourages other colluders from participating in the collusion, because the payoff $\hat{u}_i^o(n)$ they get in the existence of whistleblower becomes
\begin{equation}
    \hat{u}_i^o(n) = -\alpha_i\mathcal{S}\phi(\alpha_i)  <  \phi(\alpha_i)R_B - c_T - c_V/2 + R_C = u_i^d(n)
\end{equation}
 As a result, joining and obeying the lazy collusion is not a Nash equilibrium.

 Then we discuss the impact of the whistleblower scheme and the elimination of the lazy equilibrium on \dcgame. We denote the strategy profile in $\dcgame$ as $\{a, n_c, w\}$, where $a\in\{o,b,c\}$ is the action chosen by the leader, $n_c$ is the number of watchtowers in the collusion, and $w\in \{\true,\false\}$ represents whether there exists a whistleblower. First, all strategies with $w = \false$ are not Nash equilibria by the same analysis in Theorem~\ref{thm:dc}. Next, since the existence of a whistleblower will not lower the payoff of a leader who chooses to obey, betrayal is still strictly dominated by obedience. Then we consider the strategy where a whistleblower exists in the collusion group. In this case, if the leader opts to cheat, it's always better to switch to obedience. However, if the leader chooses to obey, the whistleblower is better off choosing not to report. Consequently, there doesn't exist any pure strategy in the subgame \dcgame that is a Nash equilibrium.
\end{proof}

The introduction of a whistleblower protocol changes the payoff dynamics of lazy collusion, as any colluder can expose the collusion for a higher payoff. Knowing this, the leader may not choose to initiate lazy collusion in the first place and the full game reach back the state where diligent strategy is the only Nash equilibrium. 

\begin{corollary}
    In \podgame\ that allows lazy and diligent collusion, if whistleblower contract exists, there is a unique equilibrium that all watchtowers are independently diligent.
\end{corollary}

\paragraph*{Cryptoeconomic security and parameter selection.} 
To provide a clear benchmark for evaluation and decision-making for different security needs, we discuss how to choose parameters that ensure both cryptoeconomic security and compatible incentives.

First, we normalize the execution transaction costs to $c_T = 1$. As an example, assume there are $n=10$ watchtowers with equal stakes; then $\phi(\alpha_i) = \phi(\alpha_0) \simeq 0.2$ (Eq.~\ref{eq:cond}) when $\theta = 0.9$. Eq.~\ref{eq:cond1} then gives the bound on rewards: $R_B > 5$ and $R_C>1$, which are affordable as the normal-path incentives the rollup operator needs to provide. Eq.~\ref{eq:cond2} calculates the minimum stake $\alpha_0 \mathcal{S} > 100009$, assuming that $c_V=100000 \gg c_T$.

To induce lazy collusion, a leader in \lcgame will set the stake $t > 18514$ according to Eq.~\ref{eq:t2}. In \lcgame, for any $t \ge 0$, the condition in Theorem~\ref{thm:dc} always holds. To eliminate all collusion strategies, according to Eq.~\ref{eq:whistleblower}, the reward for the whistleblower should be set to $R_w > 120572$. 

If we increase the number of watchtowers to $n = 100$, we see that the bounty $R_B$ increases accordingly to at least $44$, which is still low. The minimum stake required for each watchtower does not change significantly, and $R_w$ decreases to $102281$. 

Under the rational adversary assumption, our protocol guarantees a unique pure strategy Nash equilibrium where all watchtowers are diligent. Beyond this, cryptoeconomic security requires that when an attack occurs, the cost of launching the attack exceeds the maximum profit. Therefore, when designing the actual parameters (e.g. $n$, $\mathcal{S}$) for a practical system, we can utilize signals on the inherent value of transactions~\cite{deb2024stakesure} to adapt the security requirements of watchtowers. For example, considering the current average transaction fee on Ethereum is approximately \$3 and the average L2 batch size is around 200, we simplify the model by assuming all transactions are executed on L1 to resolve disputes, which incurs the $c_V \approx \$600$ and $c_T \approx \$0.006$. (Using other more complex dispute resolution methods can reduce this cost.) Then the normal path reward for watchtowers can be estimated as:
\[
R_B > 5 \times \$0.006 = \$0.03 \text{ per batch}.
\]
Next, we calculate the number of batches produced by some L2 chains per year. For example, Optimism processes approximately 400,000 transactions per day, translating to about 2,000 batches per day or 700,000 batches per year. Given the probability that a watchtower wins the bounty is 0.2 and the annual percentage yield (APY) from external investment vehicles is around 6\%, the condition to incentivize stakers with the estimated return rate is:
\[
\frac{(5-1) \times 0.006 \times 0.2 \times 700{,}000}{600 + \text{tx value}} > 6\%.
\]
In other words, a watchtower would be willing to secure approximately \$56,000 worth of transactions. And if the application aims to incentivize watchtowers to secure higher value transactions, the rewards should be increased accordingly.

Additionally, the minimum stake that each watchtower needs to post should include the transaction value. Notably, $n$ determines the normal operational overhead of our protocol, which does not directly determine the security. It can be chosen with a trade-off between stake decentralization and operational cost.

\section{Implementation and Evaluation} \label{sec:impl}
\subsection{System design} \label{sec:system}

The protocol described so far is a general framework for proving diligence in a computing platform with anytrust guarantees. We describe an implementation of these proofs subsequently in the context of optimistic rollups (ORs) as the compute platform and Eigenlayer~\cite{eigenlayer} as the underlying staking platform under a non-collusion setting. The complete system is termed a watchtower network and serves the following implementation goals:
\begin{enumerate}
    \item {\em Low compute overhead}: Watching an OR state involves executing all transactions of the rollup. Overhead is termed as any resource cost on top of the bare minimum rollup state execution. Our implementation minimizes this overhead to lower a watchtower's resource costs.  
    \item {\em Modular implementation}: The rollup ecosystem has a lot of tech stacks for full nodes ranging from general OP-stack to specialized implementations for DeFi, such as LayerN. Our modular implementation can be used on any rollup stack with minimal modifications. 
    \item {\em Low gas fees:} Large gas fees on settlement layers such as Ethereum can make watching prohibitively expensive. Our implementation scales down L1 gas costs and makes it an adjustable feature for the rollup. 
\end{enumerate}

The implementation is split across the functional domains of the rollup, settlement, payment, and staking layers. For simplicity, we assume the settlement, payment, and staking layers sit on the same ledger. However, it can be easily expanded to independent networks if desired. Figure \ref{fig:implementation_overview} shows the binding of the functional components with the two architectural components: the Watchtower client running on a server and smart contracts running on Ethereum. We outline the details of the two architectural components below.
Our implementation draws from the modules in section \ref{sec:system} to build a watchtower on the Optimism Bedrock stack \cite{The_Optimism_Collective_The_Optimism_Monorepo}. We test the implementation on the Ethereum Goerli testnet to watch OP-goerli and Base-goerli. We evaluate the system as per our quantitative implementation goals as described in section \ref{sec:system}: {\em Compute overhead}, and {\em gas costs}. We adapt the modules to fit the existing rollup stacks and deploy them using an update-optimized architecture to make them evolve with the rollup ecosystem. We go over these details in Appendix~\ref{sec:system-details}. 

\begin{figure}
    \centering
    \includegraphics[width=0.8\textwidth]{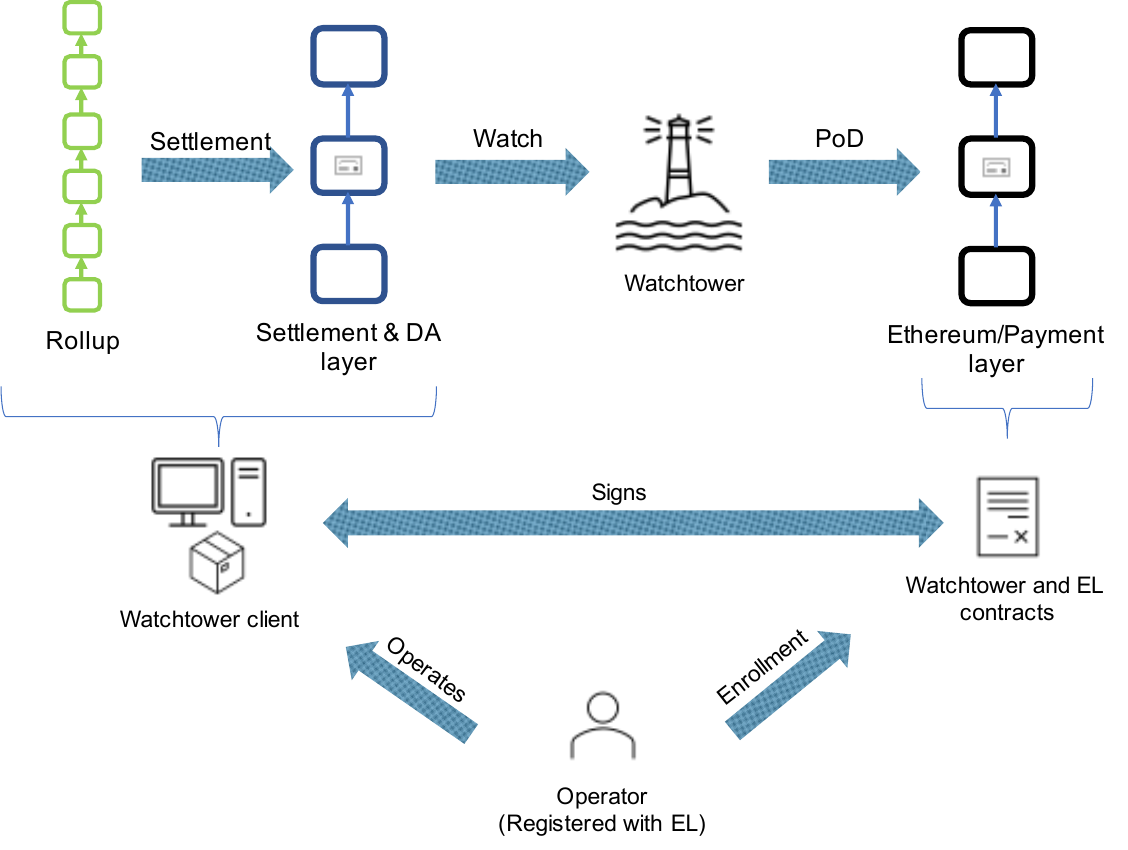}
    \caption{Watchtower client executes the rollup and observes the commitments on settlement layer, it posts bounty and flags on the payment/stake layer}
    \label{fig:implementation_overview}
\end{figure}

\subsection{On-chain Contracts}
The contracts are written in Solidity and deployed on Ethereum Goerli via a UUPS proxy architecture \cite{uupsEIP} to enable future updates. We deploy three contracts derived from section \ref{sec:system}: {\em OperatorRegistry, BountyManager}, and {\em AlertManager}. The deployed contracts can be found at \cite{deployed_contracts}.

BountyManager contract implements a hash minimum across watchtowers to ensure a single bounty winner within an epoch. The payment pool and Rollup registry contract are replaced with a simple bounty count to measure contribution and a chain-id to point to a rollup. Dispute resolution contracts will utilize L2 fraud proofs once they evolve. 

\noindent {\bf Optimizations}: Immutable data like {\em diligenceProof} are stored as calldata variables to avoid storage costs. Proof outputs as set to a fixed size using {\em keccak256} to ensure consistency in storage requirements. Hash minima is calculated by the winning party to balance gas costs. We use {\em Mappings} instead of arrays to store hash data to reduce gas costs. We utilize audited and optimized ECDSA OpenZeppelin libraries \cite{Openzepplin_contracts} to verify the authenticity of signed proofs to reduce contract risk. The contracts emit {\em NewBountyClaimed} and {\em NewBountyRewarded} for efficient notifications to off-chain clients.  

\noindent {\bf Gas usage:} A bounty mining transaction consumes 380K gas on average as shown in figure \ref{fig:gas_usage}. As a reference, this is of the same order as a token swap operation on Uniswap. Two significant contributors to gas usage are (a) Proof storage for reward reimbursement and (b) Verifying the authenticity of proofs. This gas fees can be reduced in the future by introducing an appropriate aggregation layer for submitting bounty mining transactions. 

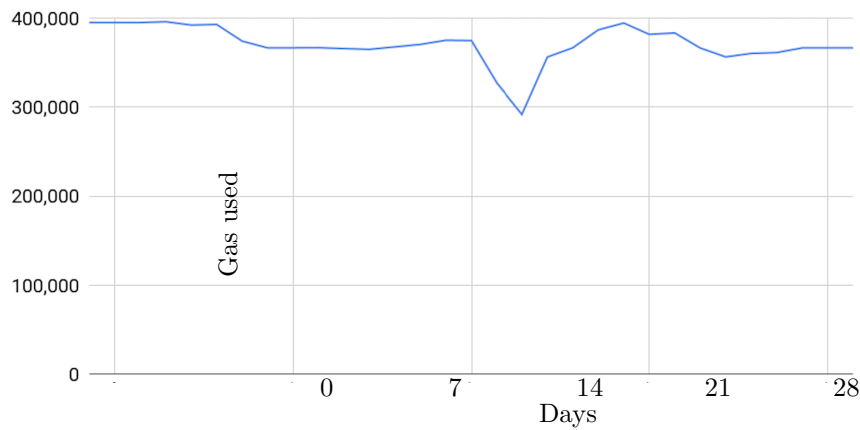
\begin{figure}
    \centering
    \includegraphics[width=0.8\textwidth]{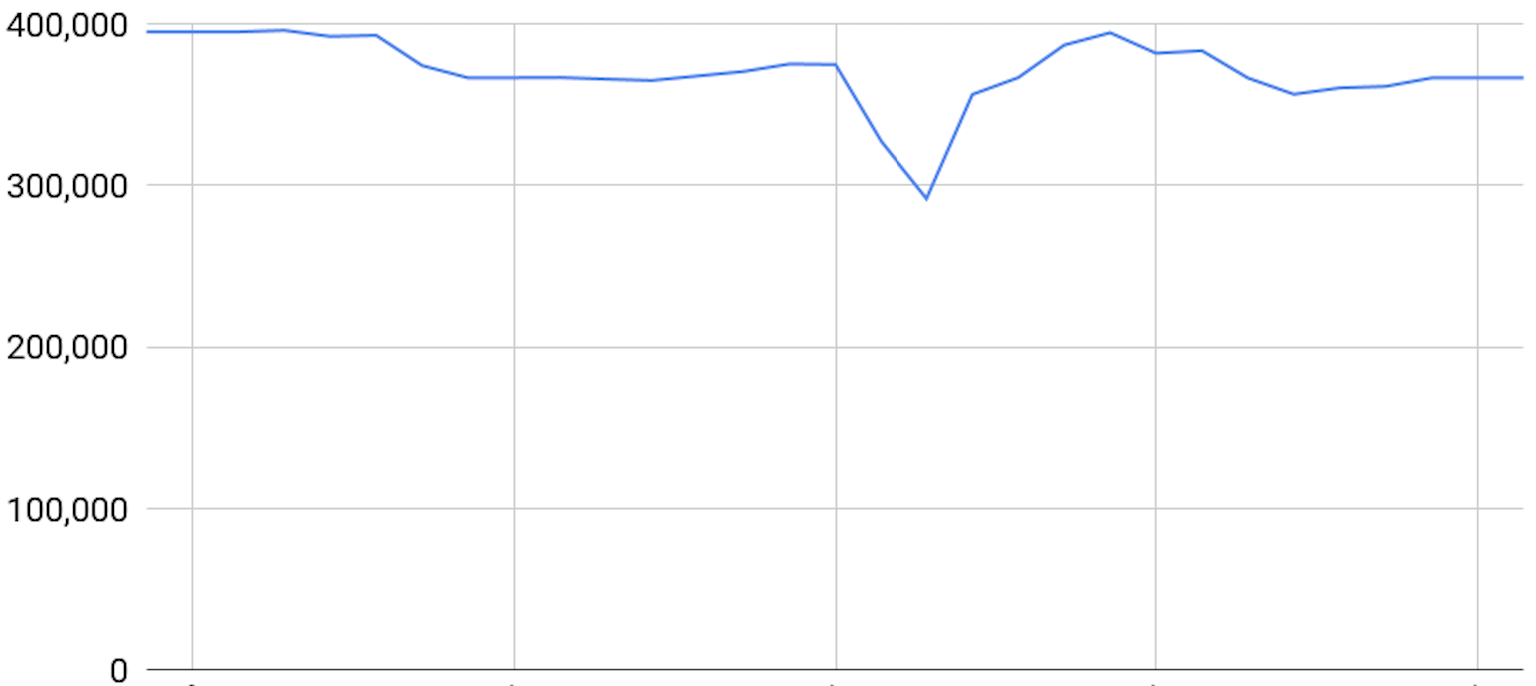}
    \put(-240,40){\rotatebox{90}{Gas used}}
    \put(-202,-5){0}
    \put(-154,-5){7}
    \put(-106,-5){14}
    \put(-58,-5){21}
    \put(-10,-5){28}
    \put(-120,-15){Days}
    \caption{Gas usage of {\em MineBounty} operation over 4 weeks of deployment}
    \label{fig:gas_usage}
\end{figure}

\subsection{Off-chain Client}
The off-chain clients are implemented to support OP-Bedrock as a rollup execution engine. Two clients implemented in Golang enable the Watchtower client. A bounty mining client contains the bounty mining and transaction generator modules. This client is supported by a State extraction client that contains the settlement layer module with an RPC connection to a full node. Alert management and reconfiguration manager modules are left as future work to be implemented once the rollup stack evolves. We describe the implementation details of the Go clients below:

\noindent {\bf Optimizations:} We employ an event-based trigger to the PoD generation. The watchtower client listens to emitted events on the {\em L2OutputOracle} Contract to receive real-time data updates from the contract; this is much more efficient than polling mechanisms. Generating execution trace requires storing the state roots after each transaction; this operation involves re-execution and hence is limited to just the penultimate block in an L2 epoch to avoid re-executing the whole epoch. 

\noindent {\bf Resource utilization:} The L2 full node and watchtower client are implemented on different machines. We run the L2 full node on a machine with 4 cores and 16GB of RAM and the watchtower client on a machine with 2 cores and 4GB RAM. Table \ref{tab:client_usage} lists the resource costs in running the watchtower client. We observe that the client consumes minimal resources natively and has a minimal resource usage overhead on the L2 node. The client has no I/O since its context is stored in memory.

\begin{table}[]
    \centering
    \begin{tabular}{@{} |c|c|c|c|c|}
        \hline
        {\bf L2 node} & {\bf CPU(\%)} & {\bf Mem(MB)} & {\bf I/0(KB/s)} & {\bf NW(KB/s)} \\
        \hline
        {\bf Before mining} & 0.7-25 & 5306 & 700 & 10-150 \\
        {\bf During mining} & 0.7-25 & 5307 & 700 & 10-150 \\
        \hline
        {\bf Watchtower client} & {\bf CPU(\%)} & {\bf Mem(MB)} & {\bf I/0(KB/s)} & {\bf NW(KB/s)} \\
        \hline
        {\bf Before mining} & 0.1 & 15 & 0 & 1 \\
        {\bf During mining} & 0.1 - 10 & 15 & 0 & 1 \\
        \hline
    \end{tabular}
    \caption{Watchtower client usage stats}
    \label{tab:client_usage}
\end{table}

We observe CPU usage burst to 10\% on the watchtower client when the rollup proposer on L1 output Oracle proposes a new block. We observe a similar burst of 25\% on the L2 goerli node when op-node sends a block to op-geth to execute the transactions and update the state. 

We can utilize this capacity for additional off-chain computing to redistribute some on-chain contract operations to an off-chain module. A proposed approach is to perform proof aggregation on-chain; this implementation is left for future work. 
\section{Discussion}
\label{sec:discussion}

Proof of Diligence ensures that watchtower network executes all transactions on the rollup diligently. Further improvements down the line will enhance security, enable generalized applications, and allow for efficient trade-offs between delay and stake. We describe such improvements below: 

\subsection{Enabling Cryptoeceonomically Secure Watchtower Applications}

The design described so far ensures that watchtowers are independently verifying transactions on rollups. The verification results can be utilized for attesting to any event on the rollup. These attestations are cryptoeconomically secured by the watchtower's stake locked with EigenLayer. We summarize the design here: 

\begin{itemize}
    \item {\bf Configurable execution event trace}: Applications can subscribe to the Watchtower network to get verifiable updates on their transactions' life cycle. The events emitted from these transactions will be added to the bounty to ensure the execution and can be challenged for cryptoeconomic security. 
    \item {\bf Application event tracing}: Watchtowers can trace the whole life cycle of transactions pertaining to an application, starting from being sequenced by the sequencer to being ordered on L1 to being asserted into the state. A different level of cryptoeconomic security will accompany each of these stages. 
    \item {\bf Dispute resolution and cryptoeconomic security}: Events pertaining to the subscribed application are attested by the watchtower and are bound to be included in the next bounty. As enforced by the proof of diligence, other watchtowers will ensure that these attestations are correct. If these attestations are exchanged in private, An application/agent consuming this attestation can contest its correctness in the future by showing a mismatch between the watchtower's attestation and the mined bounty.  
\end{itemize}

Besides applications on the rollup, the incentivized watchtower network holds significant potential for broader applications in general verifiable computing. Our future work will explore how the Proof of Diligence protocol can be extended to various domains such as AI inference~\cite{bhat2023sakshi}, cloud computing~\cite{dong2017betrayal}, and blockchain light client protocols~\cite{moshrefi2024unconditionally}. For example, to monitor and verify AI inference tasks, each watchtower can independently re-compute the inference results from the provided input data and model parameters, raising an alert if discrepancies are found. A recent work~\cite{moshrefi2024unconditionally} demonstrates the use of watchtowers as a monitoring service to secure proof of stake blockchain states for resource-limited light clients, ensuring any invalid states are detected. This approach can be extended to a wider range of blockchain applications where light clients are prevalent. Watchtowers provide a robust layer of security and accountability.

\subsection{Enhancing Security}

The current system design ensures that Proof of Diligence works under a static adversary that can form static collusion. Resistance against a stronger dynamic adversary requires assumptions on rational independence of watchtowers and the privacy of whistleblower contracts. These assumptions can be enforced through system design by enabling random rotation of watchtowers in the pool and ensuring the privacy of the whistleblower. 

{\bf Watchtower rotation}. 
The rotation of watchtowers across different rollups in the pool is essential for security against an adaptive adversary. Watchtowers can be periodically rotated in small batches across rollups in a random and staggered manner reminiscent of the cuckoo rule \cite{cuckooDHT}. The rotation can be made more efficient by utilizing  two techniques: 
(a) {\em utilizing modularity}:  we are designing an efficient reconfiguration protocol for watchtowers rotating between rollup two rollups sharing similar modules - such as two rollups running the OP stack; 
(b) {\em stateless clients}: watchtower rotation through the reconfiguration manager can be made very efficient by removing the need to transfer state. The witness chain team is developing a stateless client architecture that removes the need to download state when reconfiguring to a new rollup.  

{\bf Private Whistleblower contracts}.  
The interactions of the whistleblower with the whistleblower contract are private to ensure that they can't be used within the collusion contract. We are designing a system to ensure these inputs stay private to the collusion contract by deploying a whistleblower contract upon request post the bounty mining period. This ensures that the address of the whistleblower contract is not static and cannot be referenced in the collusion contract. Alternate design solutions include privacy-enhancing contract structures such as Aleo. 



\bibliography{sample}
\appendix


\section{System Design Details}
\label{sec:system-details}
\subsection{Watchtower Smart contracts} 


Watchtower network's security is achieved through a set of smart contracts that give out incentives and enforce slashing rules to enable honest behavior. These smart contracts utilize the Eigenlayer restaking framework for the underlying stake as its sybil resistance mechanism. Each node in the watchtower network is a registered EL operator. The operator enrolls into the 
watchtower contracts, runs the watchtower client, and calls various contracts dictated by the rollup. Watchtower network consists of the following contracts.

\begin{itemize}
    \item {\bf OperatorRegistry contract}: Registers EL operators to the watchtower network and links them to a rollup. The contract consists of the following methods:
    \begin{enumerate}
        \item {\em RegisterOperator(OperatorID, RollupID, RotationProtection)}: Registers an operator to the given rollup. The operator will be rotated if the rotation protection is turned off
        \item {\em RotateBatch(BatchCommitment)}: Updates allocation of watchtowers to rollups, does not rotate operators with rotationProtection turned on. 
    \end{enumerate}
    \item {\bf RollupRegistry contract}: Maintains the list of rollups enrolled in the network and their programmed agreement. It can be invoked by a whitelisted set of rollup managers. The contract consists of the following methods:
    \begin{enumerate}
        \item {\em RegisterRollup(ID, termParams)}: Registers the rollup if called by the whitelisted account. termParams is a vector consisting of terms of the watchtower agreement such as geographical decentralization, bounty amount, and adjusted difficulty. 
        \item {\em RefillBounty(ID, amount)}: Adds to the bounty pool used for incentivizing mining 
        \item {\em UpdateTerms(ID, termParams)}: Updates the terms of the rollup's watchtower agreement 
    \end{enumerate}
    
    \item {\bf Bounty manager contract}: Holds funds for a rollup and disperses it when a bounty is mined. The bounty mining method ensures that all nodes are given predictable rewards in expectation with low transaction costs. It consists of the following methods: 
    \begin{enumerate}
        \item {\em SetBountyDifficulty(ID)}: Internal function that sets the mining difficulty based on the terms in RollupRegistry contract and history of bounties mined 
        \item {\em MineBounty(ID, sig($r_E,r_S$))}: Called by a watchtower to claim that it has won the bounty 
        \item {\em claimBounty()}: Called by a watchtower to claim rewards, it transfers the bounty to the watchtower's account. 
    \end{enumerate}
    
    \item {\bf Alert manager contract}: Sets an alert if a rollup's state asserted on L1 does not match the executed state on L2 as claimed by a watchtower. It sets two stages of a alert - raised - followed by either confirmed or canceled. It consists of the following methods:
    \begin{enumerate}
        \item {\em raiseAlert(ID, $r_S$)}: This alert is raised either when the bounty is mined with a state root different than the one asserted on the settlement layer $r_S$ or when a watchtower manually calls this function. Sets the alert ID based on input params. 
        \item {\em confirmAlert(alertID)}: This function is called when the rollup reorgs with a different state root. The alert is set to confirm, and the watchtower that raised the alert is rewarded using the payment pool. 
        \item {\em cancelledAlert(alertID)}: Cancels a alert after social consensus time $T_s$ has passed since $r_S$ and a reorg has not happened on the rollup. The watchtower that raised the alert is slashed using the Dispute Resolution contract. 
    \end{enumerate}
    
    \item {\bf Payment pool contract}: Maintains the funds corresponding to the bounty and alerting rewards. Sets stricter access control on what accounts can call it and modify its rules. It consists of the following internal functions:
    \begin{enumerate}
        \item {\em transfer(RollupID, OperatorID)}: Transfer funds for a rollup's pool to an operator. This function is called by claimBounty and confirmAlert methods. 
        \item {\em replenishPool(RollupID)}: Adds funds to a rollup's pool, called by the UpdateTerms contract. 
    \end{enumerate}
    
    \item {\bf Dispute resolution contract}: Handles disputes concerning alerts and bounties. This contract will then wait for a veto from the watchtower committee; if not vetoed, it will call Eigenlayer's slasher contract - where it will be subject to a second veto. It consists of the following functions:
    \begin{enumerate}
        \item {\em AlertDispute(flagID)}: Raised when a cancelledAlert is issued in the alert manager contract. It checks the status of the corresponding state assertion on the rollup's settlement layer contract. 
        \item {\em BountyDispute([sigs])}: Called when a set of mined bounties with same $r_S$ have different $r_E$. All but one $r_E$ will be declared correct by the watchtower committee, and the other bounties will be canceled. 
    \end{enumerate}
\end{itemize}

These contracts are integrated with Eigenlayer's contract suite as demonstrated in figure \ref{fig:contract_implementation}. The contracts interact with watchtower clients running on a container on the operator's machine. We outline its details next.

\begin{figure*}
    \centering
    \includegraphics[width=0.9\textwidth]{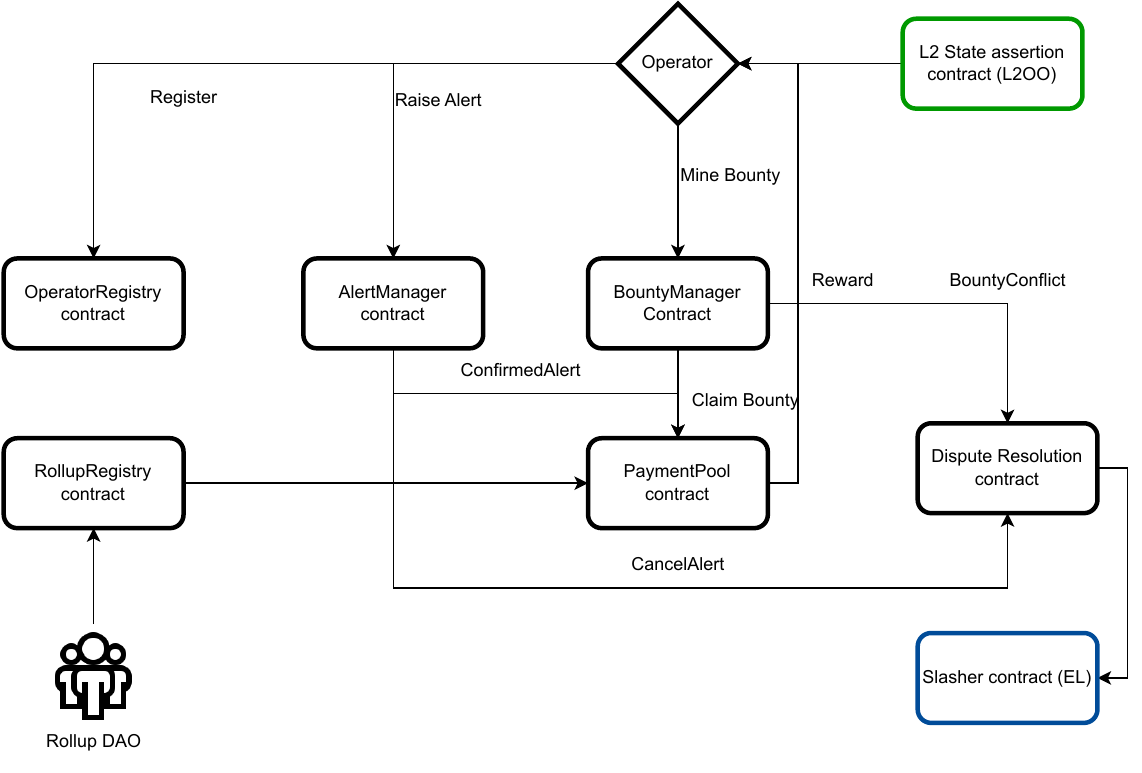}
    \caption{Integration of watchtower network contracts with L2 (green) and Eigenlayer (blue) contracts}
    \label{fig:contract_implementation}
\end{figure*}

\subsection{Watchtower Client}

Watchtower client performs three basic functions: It executes a rollup's state, observes the rollup's state assertion on its settlement layer, and mines for the bounty. It consists of the following modules that enable these functions:

\begin{itemize}
    \item {\bf Rollup execution engine}: This module subscribes to the rollup sequencer and its L2 outbox contract to fetch the sequence of transactions and execute them to compute the state locally. It maintains the following functionality:
    \begin{itemize}
        \item {\em Sequence fetcher:} Fetches the sequence of transactions. It runs in two stages: A fast forward mode - fetches sequence directly from the sequencer, and a conservative mode - fetches sequence from the L2 outbox contract on the settlement/DA layer. 
        \item {\em State root generator}: Generates state roots at a requested point in the sequence during execution. 
        \item {\em Event trace generator}: Generates event trace consisting of state roots at deterministic intermediate points of the network. 
    \end{itemize}
    \item {\bf Settlement layer light client:} This module watches the state posted by the asserter and compares output from multiple settlement layer full nodes. It runs a light client of the settlement layer to fetch the state information concerning the  assertion event. It maintains the following functionality: 
    \begin{itemize}
        \item {\em Trustless light client adapter}: Fetches API requests from multiple settlement layer full nodes to remove the trust assumption from a single full node. The operation depends on the design of the settlement layer.
        \item {\em Event watcher}: Watcher for events pertaining to L2 contracts on the settlement layer. 
    \end{itemize}
    \item {\bf Bounty miner:} This module fetches an execution trace from the rollup execution engine and checks whether the watchtower has mined a bounty. It calls the transaction generator if a bounty is mined. It maintains the following sub-modules:
    \begin{itemize}
        \item {\em Accumulator generator:} Generates a Merkle root of the execution trace fetched from the rollup execution engine 
        \item {\em Bounty checker:} Signs the hash of execution trace accumulator and state root. If the signature is below a threshold, it assembles a bounty-mining transaction and calls the transaction generator. 
    \end{itemize}
    \item {\bf Alert management module}: This module checks the state output of the rollup execution engine and the settlement layer and calls the transaction generator if there is a mismatch. It maintains the following sub-modules: 
    \begin{itemize}
        \item {\em State checker:} Fetches state assertions from settlement layer LC and State root from rollup execution engine, calls transaction generator to raise an alert if it has not been raised yet. 
    \end{itemize}
    \item {\bf Reconfiguration manager}: This module handles the rotation of the operator node to a separate rollup and state transfers to incoming nodes. It maintains the following sub-modules:
    \begin{itemize}
        \item {\em Fetch rollup modules}: Fetches new modules required for the rollup from the repo 
        \item {\em Fetch rollup state}: Fetches the latest state of the new rollup the node is assigned to 
        \item {\em Fetch network configuration}: Sets up peer connections through seed nodes 
        \item {\em Send latest state}: Send the latest state and the following transactions to a new watchtower joining the rollup. 
    \end{itemize}
    \item {\bf Transaction generator}: This module manages the keys of the operator and posts transactions to watchtower contracts as requested by the bounty miner or the alert management module. It maintains the following functionality:
    \begin{itemize}
        \item {\em Signer}: Signs the requested transaction from flag management and bounty miner after performing simple sanity checks. 
        \item {\em Alert}: Alerts the operator when a signature request is flagged by its logic; this may involve claiming an alert on a historically stable rollup. 
    \end{itemize}
\end{itemize}
\begin{figure*}
    \centering
    \includegraphics[width=0.9\textwidth]{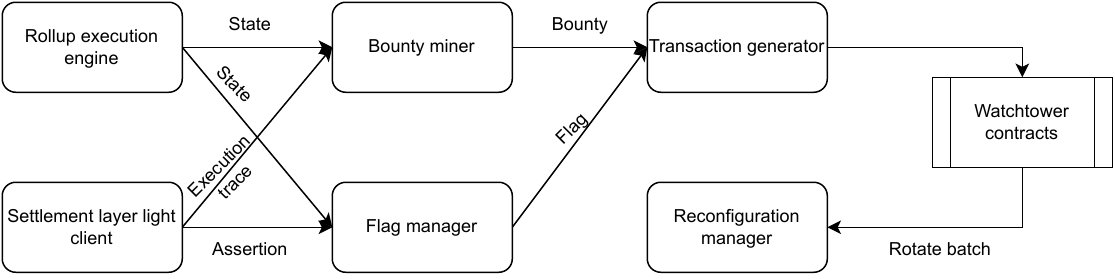}
    \caption{Watchtower client flow}
    \label{fig:watchtower_client_flow}
\end{figure*}

The modular design described above ensures that the watchtower service running PoD can quickly adapt to a new rollup. Moving from rollup A to rollup B would primarily involve switching the rollup execution engine from rollup A's full node to rollup B's full node. Moreover, the additional computation tasks needed apart from running a rollup full node are minimal and captured almost entirely under the Bounty miner module. This cost can be scaled up or down by adjusting the execution trace size used by the accumulator generator.

The high-level flow of a watchtower client is captured in figure \ref{fig:watchtower_client_flow}. The settlement layer light client observes a state assertion event; this triggers the bounty miner module to fetch the state assertion $r_S$ and execution trace $E$ from the rollup execution engine and check if a bounty is mined. If a bounty is mined, the bounty miner will call the transaction generator, which will post a transaction on the Bounty manager contract. Simultaneously, the alert management module compares the two state roots; if a discrepancy is detected, it will call the transaction generator to send a transaction to the alert manager contract. 
\end{document}